\documentclass[runningheads]{llncs}

\usepackage[T1]{fontenc}
\usepackage[utf8]{inputenc}

\usepackage{amsmath}
\usepackage{amssymb}
\usepackage{amsfonts}
\usepackage{mathtools}
\usepackage{bm}

\usepackage{graphicx}
\usepackage{enumerate}
\usepackage[hyphens]{url}
\usepackage{hyperref} 
\usepackage{comment}
\usepackage{verbatim}

\usepackage{algorithm}
\usepackage{algorithmic}
\usepackage{hyperref}
\hypersetup{
    colorlinks=true,    
    linkcolor=blue,      
    citecolor=blue,       
    filecolor=magenta,    
    urlcolor=blue       
}

\spnewtheorem{assumption}{Assumption}{\bfseries}{\itshape}

\DeclareUnicodeCharacter{2212}{-}
\begin{document}
\title{A Mathematical Programming Approach to Computing and Learning Berk--Nash Equilibria in Infinite-Horizon MDPs}
\author{Quanyan Zhu \and Zhengye Han}
\authorrunning{Q. Zhu and Z. Han}
\institute{Department of Electrical and Computer Engineering, \\
New York University, Brooklyn, NY 11201, USA \\
\email{\{qz494,zh3286\}@nyu.edu}}

\maketitle              
\begin{abstract}
We study sequential decision-making when the agent’s internal model class is misspecified. 
Within the infinite-horizon Berk--Nash framework, stable behavior arises as a fixed point: the agent acts optimally relative to a subjective model, while that model is statistically consistent with the long-run data endogenously generated by the policy itself.  We provide a rigorous characterization of this equilibrium via coupled linear programs and a bilevel optimization formulation. 
To address the intrinsic non-smoothness of standard best-response correspondences, we introduce entropy regularization, establishing the existence of a unique soft Bellman fixed point and a smooth objective.  Exploiting this regularity, we develop an online learning scheme that casts model selection as an adversarial bandit problem using an EXP3-type update, augmented by a novel conjecture-set zooming mechanism that adaptively refines the parameter space.  Numerical results demonstrate effective exploration–exploitation trade-offs, convergence to the KL-minimizing model, and sublinear regret.
\end{abstract}

\section{Introduction}

Decision-making under uncertainty is fundamental to modern networked systems \cite{li2022confluence,zhu2025revisiting,zhu2011distributed}. In practice, agents rarely possess complete knowledge of their environments and instead rely on subjective parametric models to guide decision-making. Classical learning theory suggests that, given sufficient data, models converge to the ground truth. However, in many networked and strategic settings the environment is inherently nonstationary and uncertain, evolving in response to agents’ actions \cite{li2022role,li2025computational,li2025digital,zhao2023stackelberg}. As a result, it may be difficult, or impossible, to specify and learn the environment precisely.  Model-free Reinforcement Learning (RL) \cite{sutton1998reinforcement,zhu2012hybrid} and robust control \cite{bacsar2008h} offer powerful approaches to uncertainty, but they address different objectives. Model-free RL seeks policies that maximize expected rewards without explicitly estimating a model, typically under assumptions such as Markovian dynamics and stationarity. These assumptions are often violated in multi-agent environments characterized by strategic adaptation. Robust control, in contrast, designs policies that remain feasible under worst-case uncertainty sets, which can lead to conservative behavior and may not capture endogenous belief evolution.

We adopt a perspective centered not on the objective specification of the environment, but on the agent’s internal, or mental, model of it. Rather than assuming access to the true transition kernel, we consider an agent who is committed to a restricted, potentially misspecified parametric family of conjectured models. Such restrictions may arise from considerations of interpretability, computational tractability, structural simplicity, or prior domain knowledge \cite{lei2024adapt,li2025symbiotic}. 

Within this constrained model class, the agent updates parameters through interaction with the environment, selects policies that are optimal relative to the currently maintained model, and evaluates the adequacy of the model using the data endogenously generated by its own behavior. Beliefs shape actions, actions shape observations, and observations in turn update beliefs; the learning process is therefore intrinsically closed-loop \cite{li2024conjectural}. This leads to the following central question: if an agent is restricted to a given parametric family of mental models, behaves optimally relative to that model, and the model statistically fits the long-run data induced by the agent’s own policy, does there exist a fixed-point consistency between observations, the mental model, and the induced control policy? In other words, can belief, behavior, and endogenous data generation be mutually self-sustaining under model misspecification?

 \vspace{-4mm}
\subsubsection{Main Contribution of This Work} To systematically analyze this problem, we begin with a single-agent infinite-horizon Markov Decision Process (MDP), which serves as a conceptual foundation for more general multi-agent Markov games. Even in the single-agent case, the agent must construct and maintain an internal (mental) model of the environment. This model is updated through closed-loop interaction: the agent acts based on its current beliefs, the environment generates outcomes in response to those actions, and the resulting data informs subsequent model updates.
Although the environment is not a strategic player, it reacts to the agent’s behavior through the state-transition mechanism. The agent, in turn, best responds to the environment via its maintained model, potentially misspecified, and the policy derived from it. The interaction is therefore endogenous: beliefs shape actions, actions shape observations, and observations reshape beliefs.

Within this framework, a Berk–Nash equilibrium \cite{li2024conjectural,hammar2025adaptive,esponda2021equilibrium} is defined as a fixed point satisfying two conditions. First, subjective optimality: the agent’s policy is optimal given its maintained model. Second, statistical consistency: the maintained model minimizes the Kullback–Leibler (KL) divergence relative to the true environment under the stationary distribution induced by the agent’s own policy.
This formulation can be interpreted as a “game against nature,” and it provides a tractable proxy for localized learning in complex systems. It ensures mutual consistency between the agent’s beliefs and actions: even if the model offers a simplified representation of reality, the resulting strategy–belief pair is self-sustaining under the data generated by the agent’s behavior.

While the Berk–Nash framework provides a systematic equilibrium notion under subjective modeling, its computational aspects remain largely unexplored. Computing such equilibria is challenging because the best-response correspondence is typically non-unique and non-smooth, making fixed-point search unstable.
We emphasize that computation itself is part of the agent’s mental model, a computational mental model that determines how beliefs are represented, updated, and translated into policies. From this perspective, equilibrium must be algorithmically constructible within the agent’s representational and computational constraints.

To obtain a well-posed formulation, we introduce entropy regularization \cite{neu2017unified}, which smooths the policy landscape and guarantees uniqueness. This leads to a bilevel programming characterization: the lower level enforces subjective optimality via a linear programming (LP) representation of the MDP \cite{feinberg2012handbook,filar2012competitive}, while the upper level minimizes the long-run statistical divergence between the maintained model and the true environment. For online learning over finite but large parameter spaces, we propose an EXP3-type algorithm \cite{auer2002nonstochastic} combined with a Conjecture Set Zooming mechanism that prunes implausible models and refines promising regions. This operationalizes the computational mental model and enables approximation of the entropy-regularized Berk–Nash equilibrium with sublinear regret.

 \vspace{-4mm}
\subsubsection{Related Work}

Our work relates to the literature on equilibrium computation and learning in dynamic and stochastic games \cite{fudenberg1998theory}. 
Recent studies such as \cite{li2025computational,li2024conjectural} analyze conjectural or belief-based equilibria using iterative or algorithmic approaches. 
In contrast, we develop a mathematical programming framework that characterizes Berk--Nash equilibria \cite{esponda2016berk,esponda2021equilibrium} through coupled linear programs and bilevel optimization. 
This formulation not only yields computational procedures but also facilitates structural analysis, including convexity, regularity, dual representations, and fixed-point existence.

Berk--Nash equilibrium is closely connected to self-confirming equilibrium \cite{fudenberg1993self} and bounded rationality \cite{simon1955behavioral}. 
Like self-confirming equilibrium, it requires beliefs to be consistent with realized play rather than globally correct. 
However, it strengthens this notion by imposing Kullback--Leibler optimality within a parametric model class, thereby formalizing misspecification in the spirit of \cite{berk1966limiting}. 
The restriction to a model class captures cognitive or computational constraints while maintaining sequential rationality within that class.

Our entropy-regularized formulation connects to smooth best responses and logit equilibria \cite{mckelvey1995quantal}, as well as regularized dynamic programming \cite{geist2019theory}. 
Entropy regularization ensures uniqueness and smoothness of the best-response map, enabling tractable computation and online learning via adversarial bandits \cite{auer2002nonstochastic}.

 \vspace{-4mm}
\subsubsection{Organization of the Paper}

The remainder of the paper is organized as follows. 
Section~\ref{Sec:2} introduces the infinite-horizon Berk--Nash framework, defining subjective MDPs, KL-based statistical consistency, and the equilibrium concept. 
Section~\ref{Sec:3} presents the mathematical programming formulation, including the coupled linear programs and the joint feasibility characterization linking optimal control and belief consistency. 
Section~\ref{Sec:4} develops the entropy-regularized framework, proving existence and uniqueness of the soft Bellman fixed point and establishing smooth structural properties for computation. 
Sections~\ref{sec:exp3-case-study} and~\ref{sec:conjecture-update} describe the online learning algorithms and numerical experiments. 
Section~\ref{sec:conclusion} concludes the paper.
\section{Infinite-Horizon Berk--Nash Framework for a Finite Markov Decision Problem}
\label{Sec:2}

We present a formal framework in which a decision-maker interacts with an
unknown Markov environment while reasoning through a potentially misspecified
parametric model. The agent behaves optimally relative to her subjective mental
model, and this model must in turn be statistically consistent with the
long-run data generated by her own behavior. The fixed points of these two
requirements define an infinite-horizon Berk–Nash solution. This framework is
aligned with the formulation in \cite{esponda2021equilibrium}.

\subsection{True Markov Decision Process}

We consider a system with a \emph{finite state space} \(X = \{1,\dots,S\}\) and a \emph{finite action space} \(A = \{1,\dots,m\}\). A \emph{stationary randomized policy} is a mapping \(\pi : X \to \Delta(A)\), where \(\Delta(A)\) denotes the probability simplex on \(A\) (i.e., \(\pi(a\mid x)\ge 0\) and \(\sum_{a\in A}\pi(a\mid x)=1\)). The environment is characterized by the true transition kernel \(P(\cdot \mid x,a)\in\Delta(X)\) defined for all \((x,a)\in X\times A\), along with a one-period reward function \(r:X\times A \to \mathbb{R}\). Given a policy \(\pi\), the true controlled process evolves according to the dynamics \(a_t \sim \pi(\cdot \mid x_t)\) and \(x_{t+1} \sim P(\cdot \mid x_t,a_t)\) for \(t=0,1,\dots\). Finally, fixing a discount factor \(\beta\in(0,1)\), the infinite-horizon discounted value of policy \(\pi\) for an initial state \(x\) is given by \(V^{\pi}(x) = \mathbb{E}_{P}^{\pi} [ \sum_{t=0}^\infty \beta^t\, r(x_t,a_t) \mid x_0=x ]\).

\subsection{Subjective Parametric Models}

The decision-maker does not know the true kernel $P$. Instead, she reasons using a \emph{misspecified parametric family} of Markov kernels $\mathcal{Q} = \{ Q_\theta : \theta\in\Theta \}$, where each subjective kernel satisfies $Q_\theta(\cdot \mid x,a)\in\Delta(X)$. The parameter space $\Theta \subset \mathbb{R}^k$ is assumed to be a nonempty compact subset, and the true kernel $P$ need not belong to $\mathcal{Q}$. Given a parameter $\theta\in\Theta$ and a stationary policy $\pi$, the agent believes the process evolves according to the dynamics $a_t \sim \pi(\cdot\mid x_t)$ and $x_{t+1} \sim Q_\theta(\cdot \mid x_t,a_t)$. Expectations under this subjective model are denoted by $\mathbb{E}^{\pi}_{Q_\theta}[\cdot]$.

\begin{definition}[Subjective MDP]
Fix $\theta\in\Theta$.
The subjective Markov decision process is the tuple \( \mathcal{M}_\theta := (X,A,Q_\theta,r,\beta),\)
in which the agent evaluates policies using the discounted objective
\(V_\theta^\pi(x)
=
\mathbb{E}_{Q_\theta}^{\pi}
\!\left[
\sum_{t=0}^{\infty}
\beta^t r(x_t,a_t)
\;\middle|\; x_0=x
\right].\)
The associated optimal value function is \(V_\theta(x)
=
\sup_{\pi} V_\theta^\pi(x).\)
\end{definition}

\subsection{Subjective Optimality and the Best Response}

Because $X$ and $A$ are finite and $\beta\in(0,1)$, the subjective optimal value
function $V_\theta$ is the unique solution to the Bellman equation
\[
V_\theta(x)
=
\max_{a\in A}
\left\{
r(x,a)
+
\beta
\sum_{x'\in X}
Q_\theta(x' \mid x,a) V_\theta(x')
\right\},
\qquad x\in X.
\]

\begin{definition}[Subjectively Best-Response Policy]
Given $\theta\in\Theta$, a stationary policy $\pi$ is a \emph{best response} to
$\theta$ if for all $x\in X$,
\[
\pi(a\mid x)>0
\;\Rightarrow\;
a\in \arg\max_{b\in A}
\left\{
r(x,b)
+
\beta \sum_{x'} Q_\theta(x'\mid x,b)V_\theta(x')
\right\}.
\]
The \emph{best-response correspondence} is
\( BR(\theta)
=
\{\pi : \pi \text{ is a best response under } Q_\theta\}. \)
\end{definition}

Thus $BR(\theta)$ captures subjective sequential rationality: the agent acts
optimally relative to her own subjective model $Q_\theta$.

\subsection{Long-Run Statistical Consistency}

Let $\pi$ be any stationary policy.  
Under the true model $P$ and policy $\pi$, define the induced Markov kernel \(P_\pi(x' \mid x)
:=
\sum_{a\in A}\pi(a\mid x) P(x'\mid x,a),\)
and suppose (see Assumption~\ref{ass:ergodicity} below) that $P_\pi$ admits a unique
stationary distribution $\mu_\pi\in\Delta(X)$ satisfying \( \mu_\pi(x')
=
\sum_{x\in X}
\sum_{a\in A}
\mu_\pi(x)\,\pi(a\mid x)\,P(x'\mid x,a).\) 
The agent observes the long-run empirical frequencies of transitions under her
own behavior.  A subjective parameter $\theta$ is statistically consistent if it
best approximates the true transition kernel in the Kullback--Leibler sense.

\begin{definition}[Long-Run KL Divergence]
Given a policy $\pi$ and parameter $\theta$, the long-run KL divergence is
\[
D(\theta \mid \pi)
:=
\sum_{x\in X}
\sum_{a\in A}
\mu_\pi(x)\,\pi(a\mid x)\;
D_{\mathrm{KL}}
\!\left(
P(\cdot\mid x,a)
\;\Vert\;
Q_\theta(\cdot\mid x,a)
\right),
\]
where, for $\nu,\mu\in\Delta(X)$ with $\nu\ll \mu$, \( D_{\mathrm{KL}}(\nu\Vert\mu)
=
\sum_{x'\in X}
\nu(x')\log\frac{\nu(x')}{\mu(x')}. \)
\end{definition}

\begin{definition}[Pseudo-True Parameter Set]
For a stationary policy $\pi$, the set of Kullback--Leibler minimizing parameters is \(\Theta^\ast(\pi)
:=
\arg\min_{\theta\in\Theta}
D(\theta\mid \pi).\)
Elements of $\Theta^\ast(\pi)$ are called \emph{pseudo-true parameters} associated with $\pi$.
\end{definition}

\subsection{Infinite-Horizon Berk--Nash Solution}

\begin{definition}[Infinite-Horizon Berk--Nash Solution]
A pair $(\pi^\ast,\theta^\ast)$ consisting of a stationary policy
$\pi^\ast\in\Delta(A)^X$ and a parameter $\theta^\ast\in\Theta$ is an
\emph{infinite-horizon Berk--Nash solution (BN solution)} if:
\begin{enumerate}
\item[(i)] \textbf{Subjective Optimality}:
\(
\pi^\ast \in BR(\theta^\ast),
\)
that is, $\pi^\ast$ is a subjectively optimal policy for the subjective MDP
$\mathcal{M}_{\theta^\ast}$.

\item[(ii)] \textbf{Statistical Consistency}:
\(
\theta^\ast \in \Theta^\ast(\pi^\ast),
\)
meaning that $\theta^\ast$ minimizes the long-run KL divergence between the
true kernel $P$ and the subjective kernel $Q_\theta$ along the stationary
distribution induced by $\pi^\ast$.
\end{enumerate}
\end{definition}

Thus a BN solution is a fixed point of the correspondence \( (\pi,\theta)
\;\mapsto\;
BR(\theta)
\times
\Theta^\ast(\pi). \)

\subsection{Regularity Conditions}
\label{sec:BN-regularity}

We now impose regularity assumptions ensuring that the best-response and
pseudo-true parameter correspondences are well-behaved. Let \( \Sigma := \Delta(A)^X \)
denote the set of stationary randomized policies, endowed with the product
topology (equivalently, the Euclidean topology on the stacked vector
$(\pi(a\mid x))_{x,a}$).

\begin{assumption}[Primitives and continuity]
\label{ass:primitives}
The following hold:
\begin{enumerate}
\item[(i)] The state and action spaces $X$ and $A$ are finite, and
$\beta\in(0,1)$.
\item[(ii)] The reward function $r(x,a)$ is bounded.
\item[(iii)] The parameter space $\Theta\subset\mathbb{R}^k$ is nonempty, compact,
and convex.
\item[(iv)] For each $(x,a,x')$, the mapping
$\theta\mapsto Q_\theta(x'\mid x,a)$ is continuous on $\Theta$.
\end{enumerate}
\end{assumption}

\begin{assumption}[Ergodicity and continuity of stationary distributions]
\label{ass:ergodicity}
For every stationary policy $\pi\in\Sigma$:
\begin{enumerate}
\item[(i)] The Markov chain on $X$ with transition kernel \( P_\pi(x' \mid x)
=
\sum_{a\in A}\pi(a\mid x) P(x'\mid x,a) \)
is irreducible and aperiodic, and hence admits a unique stationary
distribution $\mu_\pi$.
\item[(ii)] The mapping $\pi\mapsto\mu_\pi$ is continuous from $\Sigma$ to
$\Delta(X)$.
\end{enumerate}
\end{assumption}

\begin{assumption}[Well-behaved KL functional]
\label{ass:KL}
For every $(x,a)\in X\times A$ and every $\theta\in\Theta$,
$P(\cdot\mid x,a)\ll Q_\theta(\cdot\mid x,a)$, so that
$D_{\mathrm{KL}}(P(\cdot\mid x,a)\Vert Q_\theta(\cdot\mid x,a))<\infty$.
Moreover, the long-run KL divergence $D(\theta\mid\pi)$ is finite and
jointly continuous in $(\theta,\pi)\in\Theta\times\Sigma$, and convex in
$\theta$ for each fixed $\pi$.
\end{assumption}

Assumption~\ref{ass:primitives} ensures that, for each $\theta$, the subjective MDP
is a standard discounted MDP with finite state and action spaces.  Under
Assumption~\ref{ass:ergodicity}, long-run frequencies are well defined and vary
continuously with $\pi$.  Assumption~\ref{ass:KL} guarantees that the KL criterion is
finite, continuous, and yields convex pseudo-true parameter sets.

\subsection{Regularity of $BR(\theta)$ and $\Theta^\ast(\pi)$}

\begin{lemma}[Properties of the pseudo-true parameter correspondence]
\label{lem:theta-star}
Under Assumptions~\ref{ass:primitives}--\ref{ass:KL}, for each
$\pi\in\Sigma$:
\begin{enumerate}
\item[(i)] The set $\Theta^\ast(\pi)$ is nonempty, compact, and convex.
\item[(ii)] The correspondence $\Theta^\ast:\Sigma\rightrightarrows\Theta$ is
upper hemicontinuous with nonempty compact convex values.
\end{enumerate}
\end{lemma}

\begin{proof}
See Appendix ~\ref{proof:lem-theta-star}.
\end{proof}

\begin{lemma}[Properties of the best-response correspondence]
\label{lem:BR}
Under Assumptions~\ref{ass:primitives}, for each $\theta\in\Theta$:
\begin{enumerate}
\item[(i)] The set $BR(\theta)$ is nonempty, compact, and convex.
\item[(ii)] The correspondence $BR:\Theta\rightrightarrows\Sigma$ is upper
hemicontinuous with nonempty compact convex values.
\end{enumerate}
\end{lemma}

\begin{proof}
    See Appendix~\ref{proof:lem-BR}.
\end{proof}

\subsection{Existence of Infinite-Horizon Berk--Nash Solutions}

We are now ready to state and prove an existence result for BN solutions.

\begin{theorem}[Existence of Infinite-Horizon BN Solution]
\label{thm:BN-existence}
Suppose Assumptions~\ref{ass:primitives}--\ref{ass:KL} hold.  Then there
exists at least one pair $(\pi^\ast,\theta^\ast)\in\Sigma\times\Theta$ that
constitutes an infinite-horizon Berk--Nash solution, i.e., \(\pi^\ast\in BR(\theta^\ast) \text{and}\quad
\theta^\ast\in \Theta^\ast(\pi^\ast).\)
\end{theorem}

\begin{proof}
    See Appendix~\ref{proof:thm-BN}.
\end{proof}

\medskip

The theorem shows that, under mild compactness, continuity, and ergodicity
assumptions, there always exists at least one fixed point at which (i) the
agent's stationary policy is subjectively optimal for some parameter
$\theta^\ast$, and (ii) this parameter is statistically consistent with the
long-run data generated by that same policy.

\section{LP Characterization of the Subjective Best Response}
\label{Sec:3}

In this section we show that the subjectively optimal value function and
best-response policy can be computed via linear programming because the state
and action spaces are finite.  The results rely on the same regularity
assumptions introduced in Sect.~\ref{sec:BN-regularity}, particularly
Assumption~\ref{ass:primitives}.

\subsection{Primal and Dual Linear Programs}

For a fixed parameter $\theta\in\Theta$, consider the subjective MDP
$\mathcal{M}_\theta=(X,A,Q_\theta,r,\beta)$.
Let $\mathbb{R}^{X}$ denote the space of real functions $v:X\to\mathbb{R}$.

\paragraph{Primal LP (Value Function Form).}
Introduce decision variables $v(x)\in\mathbb{R}$ for each $x\in X$.
The subjective optimal control problem is equivalent to the linear program
\begin{align}
\label{eq:primal-LP}
\min_{v(\cdot)}\;\;
& \sum_{x\in X} v(x)
\\[1mm]
\text{s.t.}\quad
& v(x)
\;\ge\;
r(x,a)
+
\beta\sum_{x'\in X}
Q_\theta(x' \mid x,a)\, v(x'),
\qquad
\forall (x,a)\in X\times A.
\nonumber
\end{align}
Under Assumption~\ref{ass:primitives}, the feasible set is nonempty and compact,
and the objective is linear.  Standard MDP theory implies:

\begin{lemma}[Correctness of the Primal LP]
\label{lem:primal-LP-correct}
Under Assumption~\ref{ass:primitives}, the unique optimal solution $v^\star$
of \eqref{eq:primal-LP} satisfies \(v^\star(x)=V_\theta(x) \text{for} \forall x\in X,\)
where $V_\theta$ is the subjective optimal value function.
Moreover, for each $x$, the actions that bind the Bellman inequality correspond
to subjectively optimal actions under $Q_\theta$.
\end{lemma}

\paragraph{Dual LP (Discounted Occupation Measure Form).}
Fix an initial distribution $\mu_0\in\Delta(X)$.
Define discounted occupation measures \( \eta(x,a)\ge 0, \text{ for } (x,a)\in X\times A, \)
interpreted as the expected discounted number of visits to $(x,a)$ under the
subjective model. The dual LP is
\begin{align}
\label{eq:dual-LP}
\max_{\eta(\cdot,\cdot)}\;\;
& \sum_{x\in X}\sum_{a\in A} r(x,a)\, \eta(x,a)
\\[1mm]
\text{s.t.}\quad
& \sum_{a\in A} \eta(x,a)
=
\mu_0(x)
+
\beta \sum_{x'\in X}
\sum_{a'\in A}
Q_\theta(x \mid x',a')\, \eta(x',a'),
\qquad \forall x\in X,
\nonumber
\\
& \eta(x,a) \ge 0,
\qquad \forall (x,a)\in X\times A.
\nonumber
\end{align}

\begin{lemma}[Correctness of the Dual LP]
\label{lem:dual-LP-correct}
Under Assumption~\ref{ass:primitives}, \eqref{eq:dual-LP} has a nonempty, compact,
polyhedral feasible set and attains a maximum.  If $\eta^\star$ solves
\eqref{eq:dual-LP}, then
\[
\pi_\theta(a\mid x)
=
\frac{\eta^\star(x,a)}
{\sum_{b\in A} \eta^\star(x,b)}
\quad\text{whenever }\sum_{b\in A}\eta^\star(x,b)>0,
\]
defines a subjectively optimal stationary policy \(\pi_\theta \in BR(\theta).\)
\end{lemma}

\subsection{Statistical Consistency}

Let $\pi$ be any stationary policy and let $\mu_\pi$ be its unique stationary
distribution under the true kernel $P$, guaranteed by
Assumption~\ref{ass:ergodicity}.  The long-run Kullback--Leibler divergence
between the true kernel and the subjective kernel is
\[
D(\theta \mid \pi)
:=
\sum_{x\in X}
\sum_{a\in A}
\mu_\pi(x)\,\pi(a\mid x)\;
D_{\mathrm{KL}}
\!\left(
P(\cdot\mid x,a)
\;\Vert\;
Q_\theta(\cdot\mid x,a)
\right).
\]

\begin{definition}[Pseudo-True Parameter Set]
\[
\Theta^\ast(\pi)
=
\arg\min_{\theta\in\Theta}D(\theta\mid\pi).
\]
Under Assumptions~\ref{ass:primitives}--\ref{ass:KL}, the set $\Theta^\ast(\pi)$ is
nonempty, compact, and convex, and the correspondence
$\pi\mapsto\Theta^\ast(\pi)$ is upper hemicontinuous.
\end{definition}

 \subsection{Joint Programming Characterization of Berk--Nash Equilibrium}

We now show how the dual LP representation of subjective optimality and the
KL-based statistical consistency condition can be combined into a single joint
program that characterizes Berk--Nash equilibria.  The resulting formulation
is a (generally nonlinear) feasibility problem rather than a linear program,
but it explicitly couples  the \emph{subjective} occupation measure under $Q_\theta$, and
 the \emph{true} stationary state--action frequencies under $P$.

\begin{assumption}[Finite parameter space]
\label{ass:finite-theta}
The parameter space is finite \(
\Theta = \{\theta^1,\dots,\theta^K\},
\) and for each $k\in\{1,\dots,K\}$ we denote the corresponding subjective kernel
by $Q^k := Q_{\theta^k}$.
\end{assumption}

For each $(x,a)\in X\times A$ and each model $k$, define the (true vs.\ subjective)
state-wise KL cost \(c_k(x,a)
:=
D_{\mathrm{KL}}
\bigl(
P(\cdot\mid x,a)
\;\Vert\;
Q^k(\cdot\mid x,a)
\bigr), \)
which is finite and well-defined by Assumption~\ref{ass:KL}.  Then, as noted
earlier, for a stationary policy $\pi$ with true stationary distribution
$\mu_\pi$, the long-run KL divergence can be written as
\(D(\theta^k\mid\pi)
=
\sum_{x\in X}\sum_{a\in A}
\mu_\pi(x)\,\pi(a\mid x)\,c_k(x,a). \) To obtain a linear representation of this expression, it is convenient to work
directly with the \emph{true stationary state--action frequencies}.

\paragraph{True stationary state--action frequencies.}
Let $d(x,a)\ge 0$ denote the joint stationary distribution of $(x,a)$ under the
true kernel $P$ and a stationary policy $\pi$, i.e., \( d(x,a) = \mu_\pi(x)\,\pi(a\mid x). \)
Then $d(\cdot,\cdot)$ satisfies
\begin{align}
\label{eq:true-stationary-d}
\sum_{x\in X}\sum_{a\in A} d(x,a) &= 1,\\[1mm]
\sum_{a\in A} d(x',a)
&=
\sum_{x\in X}\sum_{a\in A}
d(x,a) P(x'\mid x,a),
\qquad \forall x'\in X.
\nonumber
\end{align}
Conversely, any $d$ satisfying \eqref{eq:true-stationary-d} induces a stationary
policy (on the support of $d$) via \( \pi(a\mid x)
=
\frac{d(x,a)}{\sum_{b\in A}d(x,b)}, \text{whenever }\sum_{b\in A}d(x,b)>0. \) With this notation, the long-run KL divergence under model $k$ becomes the
\emph{linear} functional \(D_k(d)
:=
\sum_{x\in X}\sum_{a\in A} d(x,a)\,c_k(x,a).\)

\paragraph{Subjective discounted occupation measures.}
For each candidate model $k$, the dual LP \eqref{eq:dual-LP} yields the
\emph{subjective} discounted occupation measure $\eta^k(x,a)$ under $Q^k$:
\begin{align}
\label{eq:dual-LP-k}
\max_{\eta^k(\cdot,\cdot)}\;\;
& \sum_{x\in X}\sum_{a\in A} r(x,a)\, \eta^k(x,a)
\\[1mm]
\text{s.t.}\quad
& \sum_{a\in A} \eta^k(x,a)
=
\mu_0(x)
+
\beta \sum_{x'\in X}
\sum_{a'\in A}
Q^k(x \mid x',a')\, \eta^k(x',a'),
\quad \forall x\in X,
\nonumber
\\
& \eta^k(x,a) \ge 0,
\qquad \forall (x,a)\in X\times A.
\nonumber
\end{align}
By Lemma~\ref{lem:dual-LP-correct}, any optimal solution $\eta^{k,\star}$ of
\eqref{eq:dual-LP-k} induces a best-response policy $\pi^k\in BR(\theta^k)$ via
normalization \[\pi^k(a\mid x)
=
\frac{\eta^{k,\star}(x,a)}
{\sum_{b\in A} \eta^{k,\star}(x,b)} \text{ for }\sum_{b\in A}\eta^{k,\star}(x,b)>0. \]

\subsubsection{Joint Feasibility Problem for Berk--Nash Equilibrium}

We can now state a joint program that characterizes BN equilibria as feasible
points of a coupled system involving subjective discounted flows under a candidate model $k$, and
 true stationary frequencies under the induced policy.

\begin{definition}[Joint BN Feasibility Problem]
\label{def:BN-joint-program}
Under Assumptions~\ref{ass:primitives}, \ref{ass:ergodicity}, \ref{ass:KL}
and~\ref{ass:finite-theta}, consider the following system of variables:
\[
k\in\{1,\dots,K\},\quad
\eta(x,a)\ge 0,\quad
d(x,a)\ge 0,\quad
\pi(a\mid x)\in[0,1].
\]
We say that $(k,\eta,d,\pi)$ is a \emph{joint BN-feasible point} if:
\begin{enumerate}
\item[(i)] \textbf{Subjective discounted flow under model $k$:}
\begin{align}
\sum_{a\in A} \eta(x,a)
&=
\mu_0(x)
+
\beta \sum_{x'\in X}
\sum_{a'\in A}
Q^k(x \mid x',a')\, \eta(x',a'),
\qquad \forall x\in X,
\label{eq:BN-joint-subjective-flow}
\\
\eta(x,a) &\ge 0,\qquad \forall (x,a)\in X\times A.
\label{eq:BN-joint-subjective-nonneg}
\end{align}

\item[(ii)] \textbf{True stationary state--action frequencies:}
\begin{align}
\sum_{x\in X}\sum_{a\in A} d(x,a) &= 1,
\label{eq:BN-joint-true-norm}
\\
\sum_{a\in A} d(x',a)
&=
\sum_{x\in X}\sum_{a\in A}
d(x,a) P(x'\mid x,a),
\qquad \forall x'\in X,
\label{eq:BN-joint-true-flow}
\\
d(x,a) &\ge 0,\qquad \forall (x,a)\in X\times A.
\label{eq:BN-joint-true-nonneg}
\end{align}

\item[(iii)] \textbf{Consistency of the policy:}
The same stationary policy $\pi$ governs both the subjective and true
processes, in the sense that
\begin{align}
\pi(a\mid x)
&=
\frac{\eta(x,a)}{\sum_{b\in A}\eta(x,b)}
\quad\text{whenever }\sum_{b\in A}\eta(x,b)>0,
\label{eq:BN-joint-policy-from-eta}
\\
\pi(a\mid x)
&=
\frac{d(x,a)}{\sum_{b\in A}d(x,b)}
\quad\text{whenever }\sum_{b\in A}d(x,b)>0.
\label{eq:BN-joint-policy-from-d}
\end{align}
(If the denominators vanish in some state $x$, $\pi(\cdot\mid x)$ can be chosen
arbitrarily on $A$.)

\item[(iv)] \textbf{KL-minimizing parameter index:}
The selected model $k$ minimizes the long-run KL divergence:
\begin{equation}
\label{eq:BN-joint-KL-min}
\sum_{x\in X}\sum_{a\in A} d(x,a)\,c_k(x,a)
\;\le\;
\sum_{x\in X}\sum_{a\in A} d(x,a)\,c_\ell(x,a),
\qquad \forall \ell\in\{1,\dots,K\}.
\end{equation}
\end{enumerate}
\end{definition}

Note that \eqref{eq:BN-joint-subjective-flow}--\eqref{eq:BN-joint-subjective-nonneg}
are exactly the dual LP constraints for the subjective MDP under $Q^k$, while
\eqref{eq:BN-joint-true-norm}--\eqref{eq:BN-joint-true-nonneg} are the true
stationary constraints under $P$.  The equalities
\eqref{eq:BN-joint-policy-from-eta}--\eqref{eq:BN-joint-policy-from-d}
synchronize the policy driving the subjective and true processes, and
\eqref{eq:BN-joint-KL-min} picks a KL-minimizing parameter.

\begin{proposition}[Equivalence with Berk--Nash Equilibrium]
\label{prop:BN-joint-equivalence}
Under Assumptions~\ref{ass:primitives}--\ref{ass:ergodicity}--\ref{ass:KL}
and~\ref{ass:finite-theta}, the following are equivalent:
\begin{enumerate}
\item[(i)] $(\pi^\ast,\theta^\ast)$ is an infinite-horizon Berk--Nash equilibrium.
\item[(ii)] There exists a joint BN-feasible point $(k^\ast,\eta^\ast,d^\ast,\pi^\ast)$ in the
sense of Definition~\ref{def:BN-joint-program} such that $\theta^\ast=\theta^{k^\ast}$.
\end{enumerate}
\end{proposition}

\begin{proof}
    See Appendix ~\ref{prof: BN-joint-equivalence}.
\end{proof}

\medskip
Proposition~\ref{prop:BN-joint-equivalence} provides a \emph{single joint
programming} characterization of Berk--Nash equilibria when $\Theta$ is finite:
a BN equilibrium is precisely a feasible point of the coupled system
\eqref{eq:BN-joint-subjective-flow}--\eqref{eq:BN-joint-KL-min}.  In particular,
all dual LP constraints, true stationary constraints, policy consistency, and
KL-minimization conditions are encoded in one unified feasibility problem.

\subsection{Bilevel Programming Characterization of Berk--Nash Equilibrium}

We now combine subjective optimality (dual LP) and statistical consistency
(KL minimization) into a single bilevel optimization problem.  Throughout this
section we maintain Assumptions~\ref{ass:primitives}--\ref{ass:ergodicity}--\ref{ass:KL}.

\paragraph{Variables and induced policy.}
For each parameter $\theta\in\Theta$, let $\eta(\cdot,\cdot)$ denote a
discounted occupation measure for the subjective model $Q_\theta$.
Given any feasible $\eta$, define the induced stationary policy
$\pi_\eta$ by
\begin{equation}
\label{eq:pi-from-eta}
\pi_\eta(a\mid x)
=
\frac{\eta(x,a)}{\sum_{b\in A}\eta(x,b)}
\quad\text{whenever }\sum_{b\in A}\eta(x,b)>0,
\end{equation}
and choose $\pi_\eta(\cdot\mid x)$ arbitrarily on $A$ if
$\sum_{b\in A}\eta(x,b)=0$. Given $\pi_\eta$, let $\mu_{\pi_\eta}$ be its unique stationary distribution
under the true kernel $P$ (Assumption~\ref{ass:ergodicity}).  The associated
long-run KL divergence is
\begin{equation}
\label{eq:D-theta-pi-eta}
D(\theta\mid\pi_\eta)
=
\sum_{x\in X}\sum_{a\in A}
\mu_{\pi_\eta}(x)\,\pi_\eta(a\mid x)\;
D_{\mathrm{KL}}
\!\left(
P(\cdot\mid x,a)\;\Vert\;Q_\theta(\cdot\mid x,a)
\right).
\end{equation}

\paragraph{Lower-level (subjective control).}
For a fixed parameter $\theta$, the subjective best-response LP in occupation
measure form is
\begin{align}
\label{eq:lower-LP}
\max_{\eta(\cdot,\cdot)}\;\;
& \sum_{x\in X}\sum_{a\in A} r(x,a)\,\eta(x,a)
\\[1mm]
\text{s.t.}\quad
& \sum_{a\in A} \eta(x,a)
=
\mu_0(x)
+
\beta \sum_{x'\in X}
\sum_{a'\in A}
Q_\theta(x \mid x',a')\,\eta(x',a'),
\quad \forall x\in X,
\nonumber\\
& \eta(x,a)\ge 0,\qquad \forall (x,a)\in X\times A.
\nonumber
\end{align}
Let $\mathcal{S}(\theta)$ denote the set of optimal solutions of
\eqref{eq:lower-LP}.  By Lemma~\ref{lem:dual-LP-correct}, each
$\eta\in\mathcal{S}(\theta)$ induces a best-response policy
$\pi_\eta\in BR(\theta)$ via \eqref{eq:pi-from-eta}.

\paragraph{Upper-level (KL minimization).}
We define the \emph{Berk--Nash bilevel program} as
\begin{equation}
\label{eq:BN-BP}
\min_{\theta\in\Theta,\;\eta}\;
D(\theta\mid\pi_\eta)
\quad
\text{subject to}
\quad
\eta\in\mathcal{S}(\theta),
\tag{BN--BP}
\end{equation}
where the \emph{lower-level constraint} $\eta\in\mathcal{S}(\theta)$ encodes
subjective optimality (the dual LP \eqref{eq:lower-LP}), and the \emph{upper-level objective} $D(\theta\mid\pi_\eta)$ is the long-run
KL divergence defined in \eqref{eq:D-theta-pi-eta}.

\begin{proposition}[Bilevel Characterization of Berk--Nash Equilibria]
\label{prop:BN-BP}
Under Assumptions~\ref{ass:primitives}--\ref{ass:ergodicity}--\ref{ass:KL},
let $(\theta^\ast,\eta^\ast)$ solve the bilevel program \eqref{eq:BN-BP}, and
set $\pi^\ast := \pi_{\eta^\ast}$.  Then:
\begin{enumerate}
\item[(i)] $\pi^\ast \in BR(\theta^\ast)$ (subjective optimality), and
\item[(ii)] $\theta^\ast \in \Theta^\ast(\pi^\ast)$ (statistical consistency),
\end{enumerate}
so $(\pi^\ast,\theta^\ast)$ is an infinite-horizon Berk--Nash equilibrium.
Moreover, every solution to \eqref{eq:BN-BP} is a \emph{KL-minimizing} BN
equilibrium, in the sense that it minimizes $D(\theta\mid\pi)$ among all
equilibrium pairs $(\pi,\theta)$.
\end{proposition}

\section{Entropy-Regularized Berk--Nash Framework for Infinite-Horizon Markov Decision Problems}
\label{Sec:4}

We consider an infinite-horizon discounted Markov decision process with finite state space $X=\{1,\dots,S\}$, finite action space $A=\{1,\dots,m\}$, discount factor $\beta\in(0,1)$, reward function $r:X\times A\to\mathbb{R}$, and true transition kernel $P(\cdot\mid x,a)$.  The decision-maker entertains a parametric family of subjective kernels $\mathcal Q=\{Q_\theta(\cdot\mid x,a):\theta\in\Theta\}$, where $\Theta\subset\mathbb{R}^d$ is compact and the true kernel $P$ need not belong to $\mathcal Q$. Under parameter $\theta$, the agent believes that $x_{t+1}\sim Q_\theta(\cdot\mid x_t,a_t)$ and chooses a stationary policy $\pi(a\mid x)$. To ensure uniqueness and smoothness of the best response, we introduce entropy regularization with temperature $\lambda>0$. Given $\theta$, the agent solves
\begin{equation}
\label{eq:soft-mdp}
\max_{\pi}
\;\mathbb{E}^{\pi}_{Q_\theta}\!\left[
\sum_{t=0}^{\infty}\beta^t\bigl(r(x_t,a_t)-\lambda\log\pi(a_t\mid x_t)\bigr)
\right].
\end{equation}

For $v\in\mathbb{R}^X$, define the soft Bellman operator
\[
(\mathcal T_{\theta,\lambda}v)(x)
=
\lambda\log\!\left(
\sum_{a\in A}
\exp\!\left(
\frac{1}{\lambda}
\Bigl[
r(x,a)+\beta\sum_{x'\in X}Q_\theta(x'\mid x,a)v(x')
\Bigr]
\right)
\right).
\]

\begin{theorem}[Soft Bellman Fixed Point]
For each $\theta\in\Theta$ and $\lambda>0$, the operator $\mathcal T_{\theta,\lambda}$ is a contraction on $(\mathbb{R}^X,\|\cdot\|_\infty)$ with modulus $\beta$. Hence there exists a unique fixed point $v_{\theta,\lambda}$ satisfying $v_{\theta,\lambda}=\mathcal T_{\theta,\lambda}v_{\theta,\lambda}$.
\end{theorem}

The function $v_{\theta,\lambda}$ is the unique entropy-regularized subjective value function under $Q_\theta$. Define the soft action-value function 
$Q_{\theta,\lambda}(x,a)=r(x,a)+\beta\sum_{x'\in X}Q_\theta(x'\mid x,a)v_{\theta,\lambda}(x').$

\begin{theorem}[Softmax Best Response]
For each $\theta\in\Theta$ and $\lambda>0$, the unique entropy-regularized best response is
\(
\pi_{\theta,\lambda}(a\mid x)
=
\frac{\exp\!\left(\frac{1}{\lambda}Q_{\theta,\lambda}(x,a)\right)}
{\sum_{b\in A}\exp\!\left(\frac{1}{\lambda}Q_{\theta,\lambda}(x,b)\right)}.
\)
Moreover, the mapping $\theta\mapsto\pi_{\theta,\lambda}$ is continuous.
\end{theorem}

Entropy regularization thus renders the best-response map single-valued and smooth. For any stationary policy $\pi$, let $\mu_\pi$ denote its unique stationary distribution under $P$. The long-run Kullback--Leibler divergence is
\(
D(\theta\mid\pi)
=
\sum_{x,a}
\mu_\pi(x)\pi(a\mid x)\,
D_{\mathrm{KL}}\!\left(P(\cdot\mid x,a)\Vert Q_\theta(\cdot\mid x,a)\right).
\)
Evaluating this at $\pi_{\theta,\lambda}$ yields the entropy-regularized Berk--Nash objective $J_\lambda(\theta)=D(\theta\mid\pi_{\theta,\lambda})$. Unlike the unregularized case, the smooth dependence $\theta\mapsto\pi_{\theta,\lambda}$ allows a single-level formulation:

\begin{definition}[Entropy-Regularized Berk--Nash Parameter]
For fixed $\lambda>0$, a parameter $\theta_\lambda^\ast$ solves
\(
\theta_\lambda^\ast\in\arg\min_{\theta\in\Theta}J_\lambda(\theta).
\)
The pair $(\theta_\lambda^\ast,\pi_{\theta_\lambda^\ast,\lambda})$ is called an entropy-regularized Berk--Nash solution.
\end{definition}

\medskip
\noindent

When $\Theta=\{\theta^1,\dots,\theta^K\}$ is finite and 
$J_\lambda(\theta)$ must be estimated online, model selection can be 
cast as an adversarial multi-armed bandit problem. Each parameter 
$\theta^k$ is an arm with loss $J_\lambda(\theta^k)$.

\begin{algorithm}[H]
\caption{EXP3 for Entropy-Regularized Berk--Nash Selection}
\label{alg:BN-EXP3}
\begin{algorithmic}[1]
\STATE \textbf{Input:} learning rate $\eta>0$, exploration $\gamma\in(0,1)$, horizon $T$
\STATE Initialize weights $w_1(k)=1$ for $k=1,\dots,K$
\FOR{$t=1,\dots,T$}
    \STATE $p_t(k)=(1-\gamma)\frac{w_t(k)}{\sum_{j}w_t(j)}+\frac{\gamma}{K}$
    \STATE Sample $I_t\sim p_t(\cdot)$ and set $\theta_t=\theta^{I_t}$
    \STATE Compute $v_{\theta_t,\lambda}$ and $\pi_{\theta_t,\lambda}$
    \STATE Run $\pi_{\theta_t,\lambda}$ and obtain loss estimate $\widehat J_t$
    \STATE $\widehat\ell_t(k)=
    \begin{cases}
    \widehat J_t/p_t(I_t), & k=I_t \\
    0, & \text{otherwise}
    \end{cases}$
    \STATE $w_{t+1}(k)=w_t(k)\exp(-\eta\widehat\ell_t(k))$
\ENDFOR
\STATE \textbf{Output:} empirical distribution $(p_t)$ or best parameter estimate
\end{algorithmic}
\end{algorithm}

Standard EXP3 analysis implies sublinear regret relative to the best fixed 
$\theta^k$ in hindsight, and therefore asymptotic concentration on 
KL-minimizing parameters.

\section{Numerical Illustration}
\label{sec:exp3-case-study}

We implement Algorithm~\ref{alg:BN-EXP3} on a 3-state MDP ($X=\{0,1,2\}$) with binary actions to validate its learning dynamics. The agent chooses between a conservative action ($a=0$) and an aggressive one ($a=1$, high variance). We assume the agent uses a misspecified family $\Theta=\{\theta^1,\dots,\theta^4\}$, constructed by mixing the true kernel $P$ with uniform noise levels $\varepsilon_k \in \{0.05, 0.15, 0.30, 0.45\}$. Thus, $\theta^1$ is the best approximation, while $\theta^4$ is effectively random noise. Parameters are set to $\beta=0.95$ and $\lambda=0.1$.

\vspace{-4mm}\subsubsection{Online Model Selection Dynamics}

We run the EXP3-based selection for $T=1500$ rounds. The results show that the algorithm effectively balances exploration and exploitation to identify the KL-minimizing model. As illustrated in Fig.~\ref{fig:exp3-frequency}, after a short warm-up phase the sampling distribution concentrates on the least misspecified model $\theta^1$, which is selected in approximately $82\%$ of rounds by the end of the horizon, while the more heavily misspecified models $\theta^3$ and $\theta^4$ are selected only $4.5\%$ and $3\%$ of the time, respectively, confirming that the weight update mechanism penalizes high-KL candidates. 

Figure~\ref{fig:exp3-loss} shows that the instantaneous loss $\widehat{J}_t$ stabilizes around $\approx 0.012$ when exploiting $\theta^1$, with occasional spikes (e.g., $\approx 0.079$ when selecting $\theta^3$) corresponding to deliberate exploration. The running average loss converges to $0.0122$, approaching the theoretical lower bound within the model class, thereby demonstrating sublinear regret relative to the best fixed parameter in hindsight.
\begin{figure}[t!]
    \centering
    \begin{minipage}{0.48\linewidth}
        \centering
        \includegraphics[width=\linewidth]{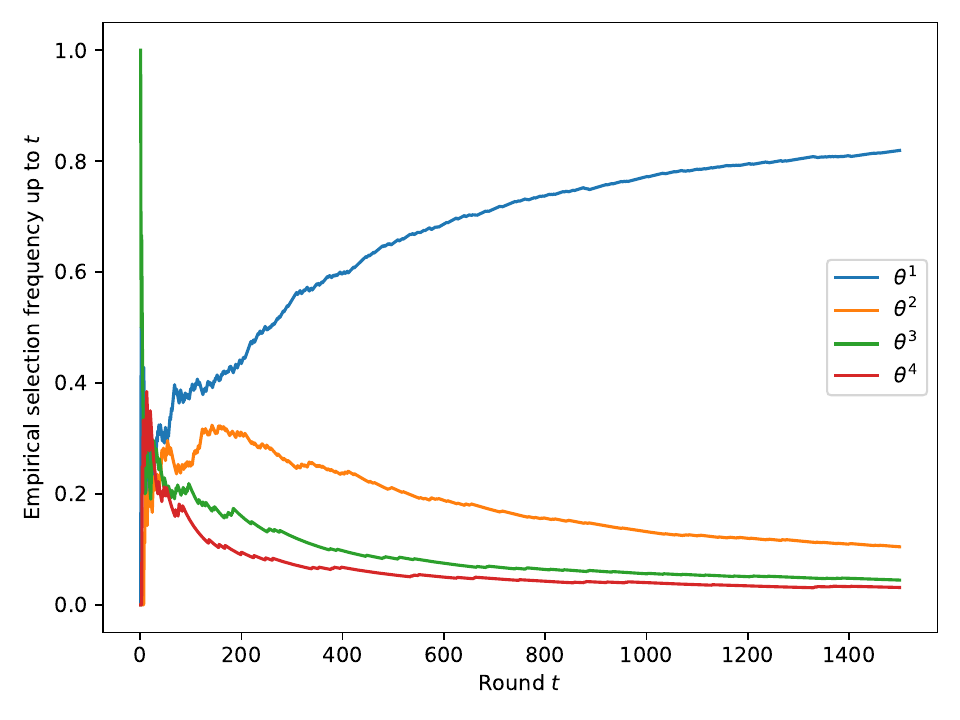}
        \vspace{-0.7cm} 
        \caption{Empirical selection frequencies. The algorithm concentrates 82\% of decisions on $\theta^1$.}
        \label{fig:exp3-frequency}
    \end{minipage}
    \hfill
    \begin{minipage}{0.48\linewidth}
        \centering
        \includegraphics[width=\linewidth]{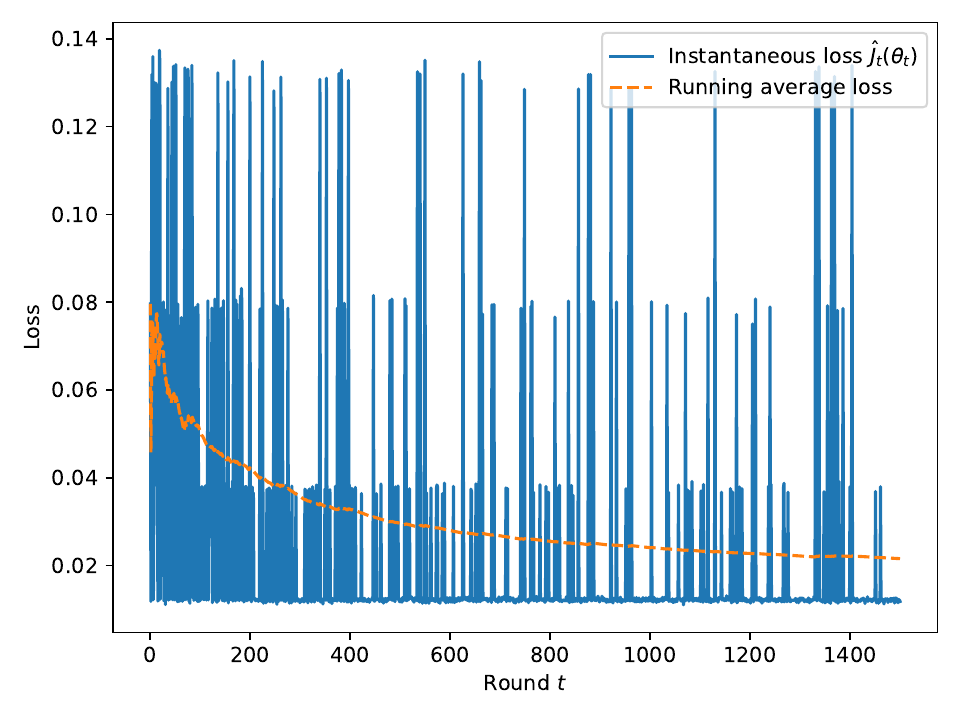}
        \vspace{-0.7cm}
        \caption{Instantaneous loss $\widehat{J}_t$ and running average. Convergence to BN loss (0.012) observed.}
        \label{fig:exp3-loss}
    \end{minipage}
    \vspace{-0.1cm} 
\end{figure}

\begin{figure}[t!]
    \centering
    \begin{minipage}{0.48\linewidth}
        \centering
        \includegraphics[width=\linewidth]{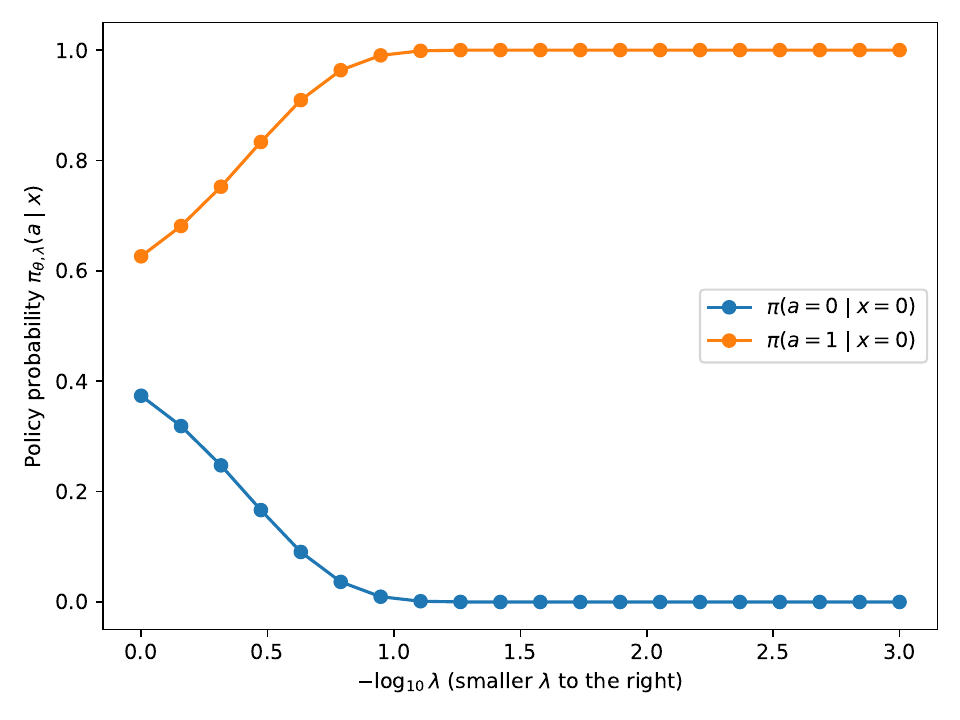}
        \vspace{-0.8cm}
        \caption{Policy $\pi_{\theta^1,\lambda}(\cdot|0)$ vs.\ $\log_{10}\lambda$. Transition from uniform to deterministic.}
        \label{fig:lambda-policy}
    \end{minipage}
    \hfill
    \begin{minipage}{0.48\linewidth}
        \centering
        \includegraphics[width=\linewidth]{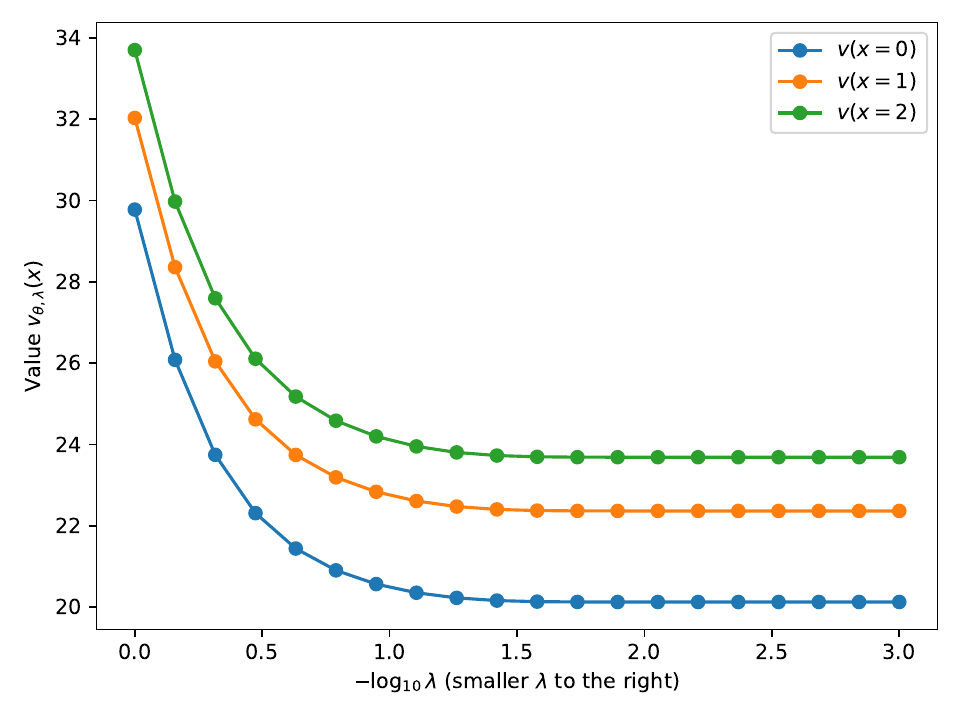}
        \vspace{-0.8cm} 
        \caption{Value $v_{\theta^1,\lambda}$ vs.\ $\log_{10}\lambda$. Values converge smoothly to the unregularized limit.}
        \label{fig:lambda-value}
    \end{minipage}
    \vspace{-0.2cm}
\end{figure}

\vspace{-4mm}\subsubsection{Sensitivity to Entropy Regularization}

We further examine how the temperature parameter $\lambda$ shapes the equilibrium associated with the best-fitting model $\theta^1$, highlighting its role as a homotopy parameter. As shown in Fig.~\ref{fig:lambda-policy}, the optimal policy probability at a reference state ($x=0$) varies smoothly with $\lambda$. For large $\lambda$, the entropy term dominates and induces a nearly uniform policy ($\pi \approx 0.5$), reflecting strong exploration. As $\lambda$ decreases toward zero, the policy becomes increasingly decisive, undergoing a sharp transition and ultimately converging to the deterministic best response of the unregularized problem. Correspondingly, Fig.~\ref{fig:lambda-value} shows that the regularized value functions are lower for large $\lambda$ due to the cost of enforced exploration. As $\lambda \to 0$, the values rise and stabilize, converging to the standard Berk--Nash value. These results confirm that the entropy-regularized formulation provides a smooth approximation to the original equilibrium and recovers it arbitrarily well in the zero-temperature limit, while retaining computational tractability.

\section{Conjecture Set Update}
\label{sec:conjecture-update}

Instead of a fixed discretization, we maintain a dynamic \emph{conjecture set} $\mathcal{C}_t \subset \Theta$ that supports \emph{zooming}. This allows the agent to start with a coarse cover and progressively refine resolution in high-probability regions while pruning suboptimal parameters.

\begin{figure}[t!]
    \centering
    \begin{minipage}{0.48\linewidth}
        \centering
        \includegraphics[width=\linewidth]{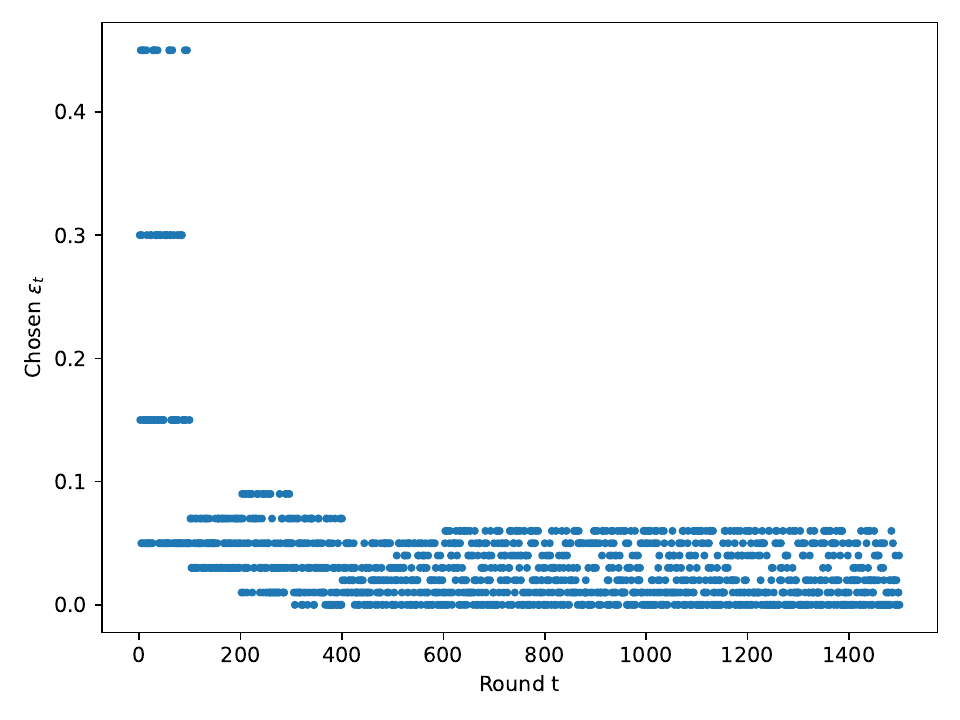}
        \vspace{-3em}
        \caption{Selected $\epsilon_t$. The algorithm rapidly concentrates on the true model region ($\epsilon \approx 0$).}
        \label{fig:zoom-theta-selection}
    \end{minipage}
    \hfill
    \begin{minipage}{0.48\linewidth}
        \centering
        \includegraphics[width=\linewidth]{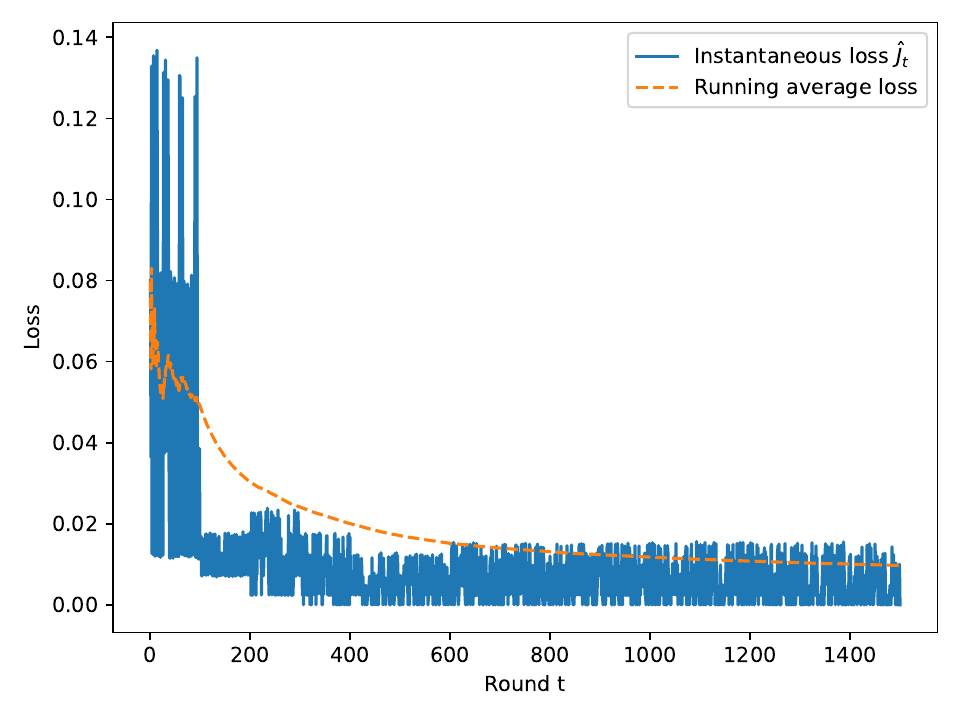}
        \vspace{-3em}
        \caption{Loss $\widehat{J}_t$. Spikes correspond to exploration; average loss converges to the minimum.}
        \label{fig:zoom-losses}
    \end{minipage}

\vspace{-2mm}\end{figure}

\vspace{-4mm} 
\subsubsection{Pruning and Refinement Operations}

At iteration $t$, for each $\theta \in \mathcal{C}_t$, we maintain a loss estimate $L_t(\theta)$ and an uncertainty indicator $U_t(\theta)$. The set evolves via two operations:

\begin{enumerate}
    \item[(a)] \textbf{Pruning:} We remove conjectures that are clearly suboptimal or sufficiently resolved. Given a threshold $\alpha_t$, we discard $\theta$ if:
    \begin{equation}
    \label{eq:prune-rule}
        L_t(\theta) > \min_{\theta' \in \mathcal{C}_t} L_t(\theta') + \alpha_t \quad \text{(suboptimal)} \quad \text{or} \quad U_t(\theta) < \delta_t \quad \text{(converged)}.
    \end{equation}

    \item[(b)] \textbf{Refinement:} We increase the resolution around promising candidates. Let $\mathcal{N}_t = \{ \theta \in \mathcal{C}_t : L_t(\theta) \le J_t^\star + \alpha_t, U_t(\theta) \ge \delta_t \}$ be the set of active, uncertain parameters. For each $\theta \in \mathcal{N}_t$, we add a local fine grid $\mathcal{R}_t(\theta) \subset B(\theta, \rho_t)$.
\end{enumerate}

The update rule is $\mathcal{C}_{t+1} = (\mathcal{C}_t \setminus \mathcal{C}_t^{\mathrm{prune}}) \cup \bigcup_{\theta \in \mathcal{N}_t} \mathcal{R}_t(\theta)$. This adaptive mechanism concentrates computational resources on the relevant subset of $\Theta$, allowing high precision without the curse of dimensionality.

\vspace{-4mm} 
\subsubsection{Numerical Illustration of Zooming}
\label{subsec:zoom-exp3-results}

We apply this adaptive EXP3 procedure to a scalar parameter $\epsilon \in [0, 0.5]$ representing model misspecification ($\epsilon=0$ is the true model; higher values imply more noise). 
Fig.~\ref{fig:zoom-theta-selection} tracks the selected $\epsilon_t$. Initially exploring a wide range, the algorithm rapidly concentrates on small values (near the true model). Occasional spikes to large $\epsilon$ reflect necessary exploration. Fig.~\ref{fig:zoom-losses} shows the corresponding loss. While exploratory pulls cause momentary spikes, the running average steadily converges to a low value, confirming that zooming effectively isolates the minimal-loss region. 

\section{Conclusion}\label{sec:conclusion}

In this work, we established a rigorous mathematical programming framework for the Infinite-Horizon Berk-Nash equilibrium, addressing the fundamental challenge of decision-making under model misspecification in complex interactive settings. By integrating linear programming characterizations with entropy regularization, we resolved the non-smoothness issues inherent in standard best-response maps, enabling the formulation of a tractable, unified objective. Our proposed online learning algorithms, specifically the EXP3-based selection augmented with adaptive conjecture set zooming, effectively bridge the gap between fixed-point theory and practical implementation, demonstrating robust convergence to KL-minimizing outcomes. Future research will extend this adaptive framework to multi-agent networked games, such as decentralized routing or distributed sensor networks, exploring how localized model misspecifications propagate through network topologies. This will further enhance the resilience of decentralized autonomous systems in unknown, highly interconnected environments.

\vspace{-4mm}

\appendix
\section{Appendix}

\subsection{Details of Algorithm ~\ref{alg:zoom-exp3}}

 \vspace{-6mm}

\noindent
\begin{algorithm}[H]
\caption{Adaptive EXP3 with Conjecture-Set Zooming}
\label{alg:zoom-exp3}
\begin{algorithmic}[1]
\STATE \textbf{Input:} initial conjecture set 
$\mathcal C_0=\{\theta_0^1,\dots,\theta_0^{K_0}\}\subseteq\Theta$,
learning rate $\eta>0$, exploration $\gamma\in(0,1)$,
horizon $T$, zoom interval $T_{\mathrm{zoom}}$,
threshold sequences $(\alpha_t)$, $(\delta_t)$,
refinement radii $(\rho_t)$

\STATE Initialize weights $w_1(k)=1$ and counters $N_0(k)=0$ for $k=1,\dots,K_0$
\STATE Set $\mathcal C_1=\mathcal C_0$, $K_1=K_0$

\FOR{$t=1,\dots,T$}

    \STATE $p_t(k)
    =(1-\gamma)\dfrac{w_t(k)}{\sum_{j=1}^{K_t} w_t(j)}
    +\dfrac{\gamma}{K_t}$

    \STATE Sample $I_t\sim p_t(\cdot)$ and set $\theta_t=\theta_t^{I_t}\in\mathcal C_t$

    \STATE Compute $v_{\theta_t,\lambda}$ and $\pi_{\theta_t,\lambda}$

    \STATE Run $\pi_{\theta_t,\lambda}$ and obtain loss estimate $\widehat J_t(\theta_t)$

    \STATE $N_t(I_t)=N_{t-1}(I_t)+1$

    \STATE $L_t(I_t)=\dfrac{N_{t-1}(I_t)L_{t-1}(I_t)+\widehat J_t(\theta_t)}{N_t(I_t)}$

    \STATE Define $U_t(k)=c\,(N_t(k)\vee 1)^{-1/2}$

    \STATE $\widehat \ell_t(k)=
    \begin{cases}
    \widehat J_t(\theta_t)/p_t(I_t), & k=I_t \\
    0, & \text{otherwise}
    \end{cases}$

    \STATE $w_{t+1}(k)=w_t(k)\exp(-\eta\,\widehat \ell_t(k))$

    \IF{$t \bmod T_{\mathrm{zoom}}=0$}

        \STATE $J_t^\star=\min_{1\le k\le K_t} L_t(k)$

        \STATE $\mathcal C_t^{\mathrm{prune}}=
        \{\theta_t^k:
        L_t(k)>J_t^\star+\alpha_t
        \ \text{or}\ 
        U_t(k)<\delta_t\}$

        \STATE $\mathcal N_t=
        \{\theta_t^k:
        L_t(k)\le J_t^\star+\alpha_t,
        \ 
        U_t(k)\ge\delta_t\}$

        \STATE For each $\theta\in\mathcal N_t$, generate
        $\mathcal R_t(\theta)\subseteq B(\theta,\rho_t)$

        \STATE $\mathcal C_{t+1}=
        (\mathcal C_t\setminus\mathcal C_t^{\mathrm{prune}})
        \cup
        \bigcup_{\theta\in\mathcal N_t}\mathcal R_t(\theta)$

        \STATE $K_{t+1}=|\mathcal C_{t+1}|$

    \ELSE

        \STATE $\mathcal C_{t+1}=\mathcal C_t$, \quad $K_{t+1}=K_t$

    \ENDIF

\ENDFOR

\STATE \textbf{Output:} empirical sampling distribution $(p_t)$ or final conjecture set $\mathcal C_T$
\end{algorithmic}
\end{algorithm}

\subsection{Proof of Lemma~\ref{lem:theta-star}}
\label{proof:lem-theta-star}
Fix $\pi\in\Sigma$. By Assumption~\ref{ass:KL}, the function
$\theta \mapsto D(\theta\mid\pi)$ is finite, continuous, and convex on
the compact convex set $\Theta$. Since $\Theta$ is compact and $D(\cdot\mid\pi)$ is continuous,
the Weierstrass theorem implies that the minimum is attained.
Hence
$\Theta^\ast(\pi)
=\arg\min_{\theta\in\Theta} D(\theta\mid\pi)$
is nonempty. Because it is the set of minimizers of a continuous
function over a compact set, it is closed; being a closed subset
of a compact set, it is therefore compact. Convexity of $D(\cdot\mid\pi)$ implies convex-valuedness of the
argmin set. Indeed, if $\theta_1,\theta_2\in\Theta^\ast(\pi)$ and
$\lambda\in[0,1]$, convexity yields
$D(\lambda\theta_1+(1-\lambda)\theta_2\mid\pi)
\le \lambda D(\theta_1\mid\pi)
+(1-\lambda)D(\theta_2\mid\pi)$.
Since both $\theta_1$ and $\theta_2$ attain the minimum,
the right-hand side equals $\min_{\theta\in\Theta} D(\theta\mid\pi)$,
so every convex combination is also a minimizer. Thus
$\Theta^\ast(\pi)$ is convex. Finally, Assumption~\ref{ass:KL} guarantees that
$D(\theta\mid\pi)$ is jointly continuous in $(\theta,\pi)$.
Because $\Theta$ is compact and does not depend on $\pi$,
Berge’s Maximum Theorem applies. It follows that the correspondence
$\pi\mapsto\Theta^\ast(\pi)$ is upper hemicontinuous with nonempty
compact values. Therefore, $\Theta^\ast(\pi)$ is nonempty, compact, convex-valued,
and the argmin correspondence is upper hemicontinuous.

\vspace{-3mm} \subsection{Proof of Lemma~\ref{lem:BR}}
\label{proof:lem-BR}

Fix $\theta\in\Theta$. Since $X$ and $A$ are finite and $\beta\in(0,1)$,
the Bellman operator $T_\theta:\mathbb{R}^X\to\mathbb{R}^X$ defined by
\[
(T_\theta V)(x)
=
\max_{a\in A}
\Bigl\{
r(x,a)
+
\beta \sum_{x'} Q_\theta(x'\mid x,a) V(x')
\Bigr\}
\]
is a contraction with modulus $\beta$ under the sup norm.
Hence, by the Banach fixed-point theorem, there exists a unique fixed point
$V_\theta$, which is the optimal value function of the subjective MDP
$\mathcal M_\theta$. For each $x\in X$, define the maximizing action set
\[
A^\ast_\theta(x)
=
\arg\max_{a\in A}
\Bigl\{
r(x,a)
+
\beta \sum_{x'} Q_\theta(x'\mid x,a) V_\theta(x')
\Bigr\}.
\]
Because $A$ is finite, the maximization is over a finite set,
so $A^\ast_\theta(x)$ is nonempty and finite.

Define the best-response correspondence
\(
BR(\theta)
=
\bigl\{
\pi\in\Sigma:
\pi(\cdot\mid x)\in\Delta(A^\ast_\theta(x))
\text{ for all } x\in X
\bigr\}.
\)
Since each $\Delta(A^\ast_\theta(x))$ is a simplex over a finite set,
it is nonempty, compact, and convex. Because $\Sigma$ is the product
of simplices over $x\in X$, it follows that $BR(\theta)$ is a
nonempty product of compact convex sets, hence compact and convex in $\Sigma$.

It remains to establish upper hemicontinuity.
The Bellman operator $T_\theta$ depends continuously on $\theta$
through the transition kernel $Q_\theta$. Because $T_\theta$
is a contraction uniformly in $\theta$, standard contraction arguments
imply that the fixed point $V_\theta$ depends continuously on $\theta$
under the sup norm. Consequently, for each $(x,a)$,
the function
\(
\theta
\mapsto
r(x,a)
+
\beta \sum_{x'} Q_\theta(x'\mid x,a) V_\theta(x')
\)
is continuous. Since $A$ is finite, the argmax correspondence
$A^\ast_\theta(x)$ is upper hemicontinuous in $\theta$ for each $x$.
Finally, $BR(\theta)$ is the product over $x\in X$
of the correspondences $\Delta(A^\ast_\theta(x))$,
and products of finitely many upper hemicontinuous correspondences
with compact values remain upper hemicontinuous.
Therefore, $BR(\theta)$ is upper hemicontinuous with
nonempty compact convex values.

\vspace{-2mm}
\subsection{Proof of Theorem~\ref{thm:BN-existence}}
\label{proof:thm-BN}
Consider the product space $K := \Sigma \times \Theta$.
Because $\Sigma$ is a product of finite-dimensional simplices
and $\Theta$ is compact and convex by assumption,
both sets are nonempty, compact, and convex.
Hence $K$ is nonempty, compact, and convex
under the product topology. Define the correspondence $F:K \rightrightarrows K$ by
$F(\pi,\theta) := BR(\theta) \times \Theta^\ast(\pi)$.
By Lemma~\ref{lem:theta-star},
$\Theta^\ast(\pi)$ is nonempty, compact, convex-valued,
and upper hemicontinuous in $\pi$.
By Lemma~\ref{lem:BR},
$BR(\theta)$ is nonempty, compact, convex-valued,
and upper hemicontinuous in $\theta$.

Because $BR$ depends only on $\theta$
and $\Theta^\ast$ depends only on $\pi$,
and both correspondences are upper hemicontinuous
with nonempty compact convex values,
their Cartesian product $F$ is upper hemicontinuous
with nonempty, compact, and convex values on $K$.
Since $K$ is compact and convex,
all the hypotheses of Kakutani’s fixed-point theorem are satisfied. Therefore, there exists $(\pi^\ast,\theta^\ast)\in K$
such that $(\pi^\ast,\theta^\ast)\in F(\pi^\ast,\theta^\ast)$.
By definition of $F$, this means
$\pi^\ast\in BR(\theta^\ast)$ and
$\theta^\ast\in \Theta^\ast(\pi^\ast)$.
Hence $(\pi^\ast,\theta^\ast)$ satisfies
the mutual best-response conditions
for an infinite-horizon Berk--Nash solution.

\vspace{-2mm}
\subsection{Proof of Proposition~\ref{prop:BN-joint-equivalence}}
\label{prof: BN-joint-equivalence}

\paragraph{$(1)\Rightarrow(2)$.}
Suppose $(\pi^\ast,\theta^\ast)$ is a Berk--Nash equilibrium.
Since the parameter set is finite, there exists $k^\ast$
such that $\theta^\ast=\theta^{k^\ast}$. Because $\pi^\ast\in BR(\theta^{k^\ast})$,
Lemma~\ref{lem:dual-LP-correct} implies that there exists
an optimal solution $\eta^\ast$ to the dual linear program
\eqref{eq:dual-LP-k} associated with $Q^{k^\ast}$ and $\pi^\ast$
such that $\eta^\ast$ satisfies the subjective flow constraints
\eqref{eq:BN-joint-subjective-flow}--\eqref{eq:BN-joint-subjective-nonneg}
and induces $\pi^\ast$ via
\eqref{eq:BN-joint-policy-from-eta}. Let $d^\ast$ denote the stationary state--action distribution
under the true kernel $P$ and policy $\pi^\ast$.
Then $d^\ast$ satisfies the true stationary flow constraints
\eqref{eq:BN-joint-true-norm}--\eqref{eq:BN-joint-true-nonneg}
and induces $\pi^\ast$ via
\eqref{eq:BN-joint-policy-from-d}. Finally, since $\theta^\ast\in\Theta^\ast(\pi^\ast)$,
the index $k^\ast$ satisfies the KL-minimality inequalities
\eqref{eq:BN-joint-KL-min}.
Therefore, $(k^\ast,\eta^\ast,d^\ast,\pi^\ast)$
satisfies all joint BN feasibility conditions.

\vspace{-2mm}\paragraph{$(2)\Rightarrow(1)$.}
Conversely, suppose $(k^\ast,\eta^\ast,d^\ast,\pi^\ast)$
is joint BN-feasible.
The subjective flow constraints
\eqref{eq:BN-joint-subjective-flow}--\eqref{eq:BN-joint-subjective-nonneg}
together with Lemma~\ref{lem:dual-LP-correct}
imply that $\pi^\ast\in BR(\theta^{k^\ast})$. The true stationary constraints
\eqref{eq:BN-joint-true-norm}--\eqref{eq:BN-joint-true-nonneg}
ensure that $d^\ast$ is the stationary state--action distribution
generated by $(P,\pi^\ast)$.
Hence the KL divergence satisfies
$D(\theta^k\mid\pi^\ast)=D_k(d^\ast)$ for every $k$.
The inequalities \eqref{eq:BN-joint-KL-min}
then imply that $\theta^{k^\ast}$ minimizes
$D(\theta^k\mid\pi^\ast)$ over $k$,
that is, $\theta^{k^\ast}\in\Theta^\ast(\pi^\ast)$.
 Thus $\pi^\ast\in BR(\theta^{k^\ast})$
and $\theta^{k^\ast}\in\Theta^\ast(\pi^\ast)$,
so $(\pi^\ast,\theta^{k^\ast})$
is a Berk--Nash equilibrium.

\vspace{-2mm}
\bibliographystyle{splncs04}  
\bibliography{refs}
\end{document}